\documentclass[a4paper,12pt]{amsart}
\usepackage{amsmath}
\usepackage{amsfonts}
\usepackage{amssymb}
\usepackage{listings}
\usepackage{mathtools}
\usepackage{hyperref}
\usepackage{amsthm}
\usepackage{verbatim}
\usepackage{bbm}

\usepackage[
]{todonotes}

\usepackage{geometry}
\usepackage{graphicx}

\usepackage[T1]{fontenc}
\usepackage[utf8]{inputenc}
\usepackage{lmodern}
\usepackage{algorithm}
\usepackage{algorithmic}
\usepackage{array}
\usepackage{multirow}
\usepackage{booktabs}

\usepackage{subcaption}
\newcommand{\rset}{\mathbb{R}}

\newcommand{\TR}{\mathrm{Tr}\,} %
\newcommand{\Hphys}{\mathcal{H}_{\mathrm{phys}}}
\newcommand{\Pphys}{\Pi_{\mathrm{phys}}}
\newcommand{\Qphys}{Q_{\mathrm{phys}}}
\newcommand{\Ephys}{E_{\mathrm{phys\text{-}proj}}}
\newcommand{\epsphys}{\epsilon_{\mathrm{phys}}}
\newcommand{\EFCI}{E_{\mathrm{FCI}}}
\newcommand{\pHEA}{\mathrm{pHEA}}
\newcommand{\ket}[1]{|#1\rangle}
\newcommand{\bra}[1]{\langle #1|}
\newcommand{\braket}[1]{\langle #1\rangle}

\numberwithin{figure}{section}
\numberwithin{table}{section}
\numberwithin{equation}{section}
\newtheorem{theorem}{Theorem}[section]

\newtheorem{proposition}[theorem]{Proposition}

\newtheorem{definition}[theorem]{Definition}
\newtheorem{remark}[theorem]{Remark}

\author[Y.-C. Chen]{Yuan-Chieh Chen}
\address{Department of Mathematical Sciences, University of Delaware, Newark, DE 19716, USA}
\email{ycchen@udel.edu}

\date{}

\begin{document}
\title[]{Energy Is Not Enough: Leakage-Aware Ansatz Criterion for Variational Quantum Eigensolvers}
\maketitle

\begin{abstract}
In molecular variational quantum eigensolver
(VQE) calculations for a fixed electron-number sector, the
projected-sector energy alone does not certify that a prepared
state is physically valid. 
We show that a state can have zero projected-sector error while almost all of its probability mass lies outside the target particle-number sector. 
We therefore evaluate ansatz families using three quantities: projected-sector error, leakage outside the physical sector and parameter count. 
Our noiseless experiments show that particle-preserving UCC-style families eliminate leakage induced by ansatze, but do not guarantee ground-state accuracy.
In the parameter-budget comparison experiments, the one-shot ranking rules select identical generator sequences and yield identical errors.
The adaptive hybrid procedure achieves lower mean projected-sector errors at selected budgets.
Using established Kraus invariant-subspace criteria, we also express channel-induced leakage through a positive operator and discuss global and local depolarizing noise analytically.
This gives a practical diagnostic framework for comparing ansatze beyond energy alone.

\end{abstract}

\tableofcontents


%
\textbf{Keywords. } variational quantum eigensolver, ansatz design,
particle-number leakage, symmetry verification, quantum error mitigation, NISQ.

\section{Introduction}

The variational quantum eigensolver \cite{peruzzo2014} approximates the molecular ground-state energy of a molecular Hamiltonian $H$ by preparing a parameterized state $\ket{\psi(\theta)}=U(\theta)\ket{\psi_{\mathrm{HF}}}$ from an ansatz $U(\theta)$ and a Hartree--Fock initial state $\ket{\psi_{\mathrm{HF}}}$, and minimizing $\bra{\psi(\theta)}H\ket{\psi(\theta)}$ with a classical optimizer. 
It's one of the most studied algorithms in the Noisy Intermediate-Scale Quantum (NISQ) era, for example, it's reviewed in \cite{mcardle2020,tilly2022}. 
For a molecule with $m$ spin orbitals the fermionic Fock space decomposes into fixed-particle-number sectors
\begin{equation}\label{eq:sectors}
\mathcal{F}=\bigoplus_{n=0}^{m}\mathcal{H}_n,
\qquad
\mathcal{H}_n=\ker(N-nI),
\end{equation}
where $N$ is the particle-number operator.
For a target electron number $N_e$, the physically valid sector is $\Hphys=\ker(N-N_eI)$.
Since the molecular Hamiltonian conserves particle number, $[H,N]=0$, each sector is invariant under $H$. 
Hence the chemically meaningful quantity is the lowest eigenvalue within $\Hphys$ rather than over the full Hilbert space.

An ansatz need not preserve the target particle
number. 
Moreover, an energy evaluated after
projection onto $\Hphys$ describes the normalized
projected state, not the full prepared state.
Let $P$ denote the orthogonal projector onto
$\Hphys$.
For a normalized prepared state
$\ket{\psi}$, leakage is 
$$\Lambda=1-\bra{\psi}P\ket{\psi}.$$
A small projected-sector energy error does not imply small leakage by itself.

This motivates reporting a VQE result not as an energy alone but as the triple
\begin{equation}\label{eq:triple-intro}
\bigl(\epsphys,\ \Lambda,\ \#\theta\bigr)
\end{equation}
of projected-sector error, leakage outside the sector, and parameter count. Here $\epsphys$ is measured relative to the target-sector ground-state energy when that reference is known.

Two structural results explain why projected-sector error must be distinguished from leakage.
Ansatze based on excitation generators satisfy $\Lambda \equiv 0$ identically, which is standard and is recorded here as
Proposition~\ref{prop:preserve} because it is what makes the UCC-style families a controlled reference.
Conversely, Theorem~\ref{thm:mask} shows that a state can have exactly zero projected-sector error while its leakage is arbitrarily close to one.

Leakage from a computational subspace has been characterized by Wood and Gambetta \cite{wood2018}.
In our setting, the relevant subspace is the fixed-particle-number sector $\Hphys$.
Streif et al.\ analyze this form of leakage under local depolarizing noise \cite{streif2021}.

In VQE, symmetry verification and post-selection restrict energy estimation to a fixed symmetry sector \cite{bonet2018}.
Huggins et al.\ incorporate particle-number information into energy measurements, enabling post-selection on the target electron number \cite{huggins2021}.
Symmetry-projected VQE combines projection directly into the variational objective through the normalized projected-energy quotient \cite{yen2019projectors}.
The distinction relevant here is between the energy of the normalized projected state and the leakage of the prepared state. 
The former does not determine the latter.
Building on these established concepts, we report projected-sector error, leakage, and parameter count to distinguish conditional energy accuracy from particle-number preservation at a specified ansatz size.

We compare particle-number-preserving ansatze with a non-preserving two-local ansatz optimized using raw energy, projected energy, or projected energy with a penalty.
We also examine one-shot generator ranking and
adaptive selection from an enriched excitation-plus-commutator pool.
The numerical results compare projected-sector
error, leakage, and parameter count.

A two-page abstract and poster titled ``Energy Is Not Enough: Leakage-Aware
Ansatz Criterion for VQE'' formed an earlier version of this study and were
accepted for IEEE QCE 2026. This extended manuscript provides fuller methods, attribution, reproducibility information, and limitations.

\subsection*{Organization}
Sections~\ref{sec:setup}--\ref{sec:channel} define the diagnostic and its state- and channel-level properties. Section~\ref{sec:ansatz} specifies the ansatze and selection rules, Section~\ref{sec:numerics} reports the experiments, and Sections~\ref{sec:discussion}--\ref{sec:conclusion} discuss their scope and conclusions.

\section{Setup and the diagnostic triple}\label{sec:setup}

Here $\mathcal{H}$ is a finite dimensional Hilbert space and Hamiltonian $H=H^\dagger$. Let $\Pphys$ be the nonzero orthogonal projector onto the target particle-number sector. Assume that $[H,\Pphys]=0$ and set
$\Qphys=I-\Pphys$. All definitions are stated for a positive-semidefinite density matrix $\rho$, so the same quantities serve the noiseless and the noisy settings.
\begin{definition}[Leakage]\label{def:leak}
The leakage of a state $\rho$ is
\begin{equation}\label{eq:leakage}
\Lambda(\rho)=\TR(\Qphys\,\rho).
\end{equation}
\end{definition}

\begin{definition}[Projected-sector energy and error]\label{def:eproj}
If $\mathrm{Tr}(\Pphys \rho)>0$, the projected-sector energy and the physical-sector error are
\begin{equation}\label{eq:eproj}
\Ephys(\rho)=\frac{\TR\bigl(H\,\Pphys\rho\Pphys\bigr)}{\TR(\Pphys\rho)},
\qquad
\epsphys(\rho)=\bigl|\Ephys(\rho)-\EFCI\bigr|,
\end{equation}
where $\EFCI=\lambda_{\min}(H|_{\Hphys})$ is the full configuration interaction (FCI) reference within the target sector.
\end{definition}

Note that when $\TR(\Pphys\rho)=0$, equivalently, we have $\Lambda=1$. So the projected-sector state
does not exist and both $\Ephys$ and $\epsphys$ are undefined. No energy value
is reported for that case.

The diagnostic we propose is the triple
\begin{equation}\label{eq:triple}
\bigl(\epsphys,\ \Lambda,\ \#\theta\bigr),
\end{equation}
recording target-sector accuracy, leakage, and ansatz size, where
$\#\theta$ is the number of variational parameters. 
For a VQE result to be interpretable we claim that all three should be reported together.
Projected-sector error alone does not bound leakage away from one, while parameter count records ansatz size separately from these state properties.

Under Jordan--Wigner transformation \cite{jordan1928}, leakage can be estimated
by measuring all qubits in the computational basis and
recording the fraction of outcomes with a particle
number different from the target number $N_e$.

\section{Sector preservation and the failure mode}\label{sec:theory}

\begin{proposition}[Sector preservation]\label{prop:preserve}
If $U(\theta)=\prod_k e^{\theta_kT_k}$ with $[T_k,N]=0$ for all $k$, and
$\ket{\psi_{\mathrm{HF}}}\in\Hphys$, then
$U(\theta)\ket{\psi_{\mathrm{HF}}}\in\Hphys$ for all $\theta$.
\end{proposition}

\begin{proof}
Since $[T_k,N]=0$, each $T_k$ leaves every eigenspace of $N$ invariant, hence so does $e^{\theta_kT_k}$ for every $\theta_k\in\rset$, and hence so does any product of such factors. As $\Hphys=\ker(N-N_eI)$ is precisely the $N_e$-eigenspace of $N$ and $\ket{\psi_{\mathrm{HF}}}\in\Hphys$, the image
remains in $\Hphys$.
\end{proof}

As a result, Ansatze built from excitation generators, such as UCCSD, satisfy $\Lambda \equiv 0 $ and provide a leakage-free reference. This applies equally to the UCC-style
families we used: UCCSD, SelectedUCC, RandomUCC and PCSD,  since all are drawn from the same excitation pool.

\subsection{Projection masks leakage}

\begin{theorem}[Projection masks leakage]\label{thm:mask}
Let $\ket{\phi_0}\in\Hphys$ be an FCI ground state of $H|_{\Hphys}$ and let $\ket{\eta}\in\Hphys^{\perp}$ be any normalized state. For any $\lambda\in[0,1)$ the state
\begin{equation}\label{eq:psilambda}
\ket{\psi_\lambda}
=\underbrace{\sqrt{1-\lambda}\,\ket{\phi_0}}_{\text{physical}}
+\underbrace{\sqrt{\lambda}\,\ket{\eta}}_{\text{leaked}}
\end{equation}
satisfies $\Ephys(\psi_\lambda)=\EFCI$, hence $\epsphys(\psi_\lambda)=0$, but
$\Lambda(\psi_\lambda)=\lambda$.
\end{theorem}

\begin{proof}
Since $\ket{\eta}\in\Hphys^{\perp}$, we have
$\Pphys\ket{\psi_\lambda}=\sqrt{1-\lambda}\ket{\phi_0}$. Substituting into
\eqref{eq:eproj}, the numerator is
$(1-\lambda)\bra{\phi_0}H\ket{\phi_0}=(1-\lambda)\EFCI$ and denominator is 
$(1-\lambda)$, so the ratio is exactly $\EFCI$ for every $\lambda\in[0,1)$.
Also, applying $\Qphys$ to $\ket{\psi_\lambda}$, we have $\Qphys\ket{\psi_\lambda}=\sqrt{\lambda}\ket{\eta}$, giving
$\Lambda(\psi_\lambda)=\lambda$.
\end{proof}

This is the central observation. It tells us that the projected energy $\Ephys$ only depends on the direction of the component inside the target sector, not on its norm. 
So reporting $\Ephys$ alone can be misleading by an arbitrary amount. 
For example, if $\lambda\to1$, we can see that almost all probability mass lies outside the physical sector
while the reported error remains exactly zero.
Also, we can't use mean particle number as a certificate.

\begin{remark}[Mean particle number does not generally certify sector membership]
\label{rem:meanN}
For $n$ spin orbitals and $1\leq N_e\leq n-1$,
the condition $\braket{\psi|N|\psi}=N_e$
does not imply $\Lambda(\psi)=0$.
Let $\ket{\phi_\pm}$ be normalized states in
$\mathcal H_{N_e\pm1}$, respectively.
Then
$$
\ket{\psi}=\frac{1}{\sqrt{2}} \left(\ket{\phi_-}+\ket{\phi_+}\right)
$$
satisfies $\braket{\psi|N|\psi}=N_e$ but $\Lambda(\psi)=1$.

\end{remark}

Proposition~\ref{prop:preserve} and Theorem~\ref{thm:mask} together motivate the triple \eqref{eq:triple}. The Proposition says that some ansatz families is safe by construction, and the Theorem says that for the other family without measuring the leakage $\Lambda$, energy itself is not enough.

\section{Channel-level theory}\label{sec:channel}

The diagnostic generalizes from states to noisy hardware. 
Let $P=\Pphys,~ Q=I-P$ and $\Phi_*(\rho)=\sum_\alpha K_\alpha\rho K_\alpha^{\dagger}$ be a
CPTP channel with Heisenberg adjoint
$\Phi(X)=\sum_\alpha K_\alpha^{\dagger}XK_\alpha$. 
We define the leakage operator
\begin{equation}\label{eq:leakop}
L_{\Phi,P}=P\,\Phi(Q)\,P=\sum_\alpha (QK_\alpha P)^{\dagger}(QK_\alpha P).
\end{equation}
For every state $\rho=P\rho P$ which is supported in the physical sector, we have
\begin{equation}\label{eq:leaktrace}
\Lambda_{\Phi_*,P}(\rho)=\TR\bigl(Q\,\Phi_*(\rho)\bigr)
=\TR\bigl(\rho\,L_{\Phi,P}\bigr).
\end{equation}

\begin{remark}[Relation to the leakage rate]\label{rem:l1}
Wood and Gambetta \cite{wood2018} define the leakage rate of a channel as the average of the state leakage over the computational subspace, equivalently its value on the maximally mixed state of that subspace,
$L_1(\Phi_*)=\Lambda\bigl(\Phi_*(P/d_{\mathrm{phys}})\bigr)$ with
$d_{\mathrm{phys}}=\operatorname{rank}(P)$. 
Comparing with \eqref{eq:leaktrace},
\begin{equation}\label{eq:l1}
L_1(\Phi_*)=\frac{\TR\bigl(L_{\Phi,P}\bigr)}{d_{\mathrm{phys}}}.
\end{equation}
For a fixed channel $\Phi_*$ and target projector $P$, the operator $L_{\Phi,P}$ determines leakage for every initial state supported in $\operatorname{ran}P$.
Its expectation depends on the initial state $\rho$, while the operator itself depends on both $\Phi_*$ and $P$.
When $\Phi_*$ represents a noise channel, it may also depend on the implemented
gate sequence and parameters.
\end{remark}

\begin{proposition}[Kraus no-leakage criterion]\label{prop:kraus}
$L_{\Phi,P}=0$ if and only if $QK_\alpha P=0$ for every Kraus operator
$K_\alpha$.
\end{proposition}

\begin{proof}
Note that each summand $(QK_\alpha P)^{\dagger}(QK_\alpha P)$ in \eqref{eq:leakop} is
positive semidefinite. Hence the sum vanishes if and only if every summand vanishes. Using the fact that $A^{\dagger}A=0$ if and only if $A=0$ the result is immediate.
\end{proof}

Proposition~\ref{prop:kraus} is the channel-level analogue of
Proposition~\ref{prop:preserve} where the noiseless statement asks that every generator commute with $N$, the noisy statement asks that no Kraus operator has weight from inside the sector to outside it.

\subsection{Example 1: Global depolarizing noise}

For global depolarizing noise, we have $\Phi_*(\rho)=q\rho+(1-q)I/d$ with $q\in[0,1]$. 
Taking the Heisenberg adjoint, we have $\Phi(Q)=qQ+(1-q)\TR(Q)I/d$ and hence
\begin{equation}\label{eq:glob}
L_{\Phi,P}=(1-q)\,\frac{d-d_{\mathrm{phys}}}{d}\,P,
\qquad
\Lambda_{\mathrm{glob}}=(1-q)\,\frac{d-d_{\mathrm{phys}}}{d},
\end{equation}
with $d=\dim\mathcal{H}$. Because $L_{\Phi,P}$ is proportional to $P$, the leakage is the same for every state in the sector. It is fixed by dimension and the noise probability alone, and it has no information distinguishing one ansatz from another. 

\subsection{Example 2: Local depolarizing noise and Pauli channels}

Streif et al.~\cite{streif2021} derive an exact formula
for the probability of remaining in a fixed particle-number sector under identical, independent local depolarizing noise.
Equivalently, the resulting leakage is the same for every input state within that sector. 
This doesn't hold for a general Pauli channel.

For a Pauli channel
$$
\Phi(\rho)=\sum_\alpha p_\alpha W_\alpha\rho W_\alpha^\dagger,
\qquad
W_\alpha\in\{I,X,Y,Z\}^{\otimes n},
\quad p_\alpha\geq0,\quad \sum_\alpha p_\alpha=1,
$$
substitution into \eqref{eq:leakop} gives
\begin{equation}\label{eq:pauli}
L_{\Phi,P}
=\sum_\alpha p_\alpha
(QW_\alpha P)^\dagger(QW_\alpha P).
\end{equation}
For an initially in-sector state $\rho$, the resulting
leakage is $\operatorname{Tr}(\rho L_{\Phi,P})$.

Under the Jordan--Wigner transformation,
\begin{equation}\label{eq:jw}
N=\sum_j \frac{I-Z_j}{2}.
\end{equation}
Consequently, $Z$-only strings commute with $N$ and
preserve every particle-number sector. Strings containing
$X$ or $Y$ can cause leakage, but need not do so:
for example, $X^{\otimes n}$ preserves the half-filled
sector when $n$ is even. The criterion for a string
to be capable of causing leakage from the chosen
sector is $QW_\alpha P\neq0$; its contribution for
a particular input state depends on that state.

\begin{remark}[Geometry independence]\label{rem:geom}
For a fixed Pauli channel and target projector, Equation~\eqref{eq:pauli} depends on $P$, the Pauli strings $W_\alpha$, and their probabilities $p_\alpha$, but not on the molecular Hamiltonian.
Under the Jordan--Wigner transformation, the projector onto the fixed-particle-number sector is determined by the qubit count and target electron number.
Consequently, molecular systems with the same encoding, qubit count, target electron number and Pauli channel have the same leakage operator.
Their leakage values can still be different because $\operatorname{Tr}(\rho L_{\Phi,P})$ depends on the initial state.
\end{remark}

\subsection{Parameter count and local leakage sensitivity}

Parameter count records ansatz size but does not determine gate count or channel-induced leakage.
For a specified gate implementation and noise model, changing the gate sequence can change
the circuit channel and its leakage operator.
Changing $\#\theta$ alone does not imply such
a change.

The local dependence of noiseless leakage on
the variational parameters is described by
its Hessian.
Let $\psi:\Omega\to\mathcal H$ be a normalized
$C^2$ state map on an open set
$\Omega\subset\mathbb R^{d_\theta}$, where
$d_\theta=\#\theta$.
Let $P$ be a fixed orthogonal projector,
$Q=I-P$, and $\theta_0\in\Omega$ satisfy
$Q\ket{\psi(\theta_0)}=0$.
Differentiating
$\Lambda(\theta)=
\bra{\psi(\theta)}Q\ket{\psi(\theta)}$
gives
\begin{equation}\label{eq:hessian}
\begin{aligned}
\nabla\Lambda(\theta_0)=0, \qquad
\partial_{jk}^2\Lambda(\theta_0)&=2M_{jk},\qquad
M_{jk}
=\operatorname{Re}
\langle\partial_j\psi(\theta_0)|
Q|\partial_k\psi(\theta_0)\rangle.
\end{aligned}
\end{equation}
Terms involving $Q\ket{\psi(\theta_0)}$ vanish.
Consequently, for a small real parameter
perturbation $h$, we have
\[
\Lambda(\theta_0+h)
=h^{\mathsf T}Mh+o(\|h\|^2).
\]
The matrix $M$ is positive semidefinite and
quantifies leakage to second order in the
specified parameter coordinates.
If the ansatz preserves the target sector
throughout a neighborhood of $\theta_0$,
then $Q\ket{\partial_j\psi(\theta_0)}=0$
for every $j$, and hence $M=0$.

\section{Ansatz families}\label{sec:ansatz}

This section defines the excitation generators, ansatz families and SelectedUCC ranking used in the numerical evaluations.

\subsection{Particle-number-preserving generators}

Let $a_p^\dagger,a_p$ satisfy the anticommutation relations, and let $N=\sum_p a_p^\dagger a_p$.
Using $[N,a_p^\dagger]=a_p^\dagger$, $[N,a_p]=-a_p$, and
$$ [N,AB]=[N,A]B+A[N,B], $$
by induction, we have, for
$M=a_{p_1}^\dagger\cdots a_{p_r}^\dagger
a_{q_1}\cdots a_{q_s}$,
\begin{equation}\label{eq:monomial}
[N,M]=(r-s)M.
\end{equation}
Thus every equation containing equally many creation and annihilation operators commutes with $N$.

The single- and double-excitation generators are
\begin{equation}\label{eq:gens}
\begin{aligned}
T_{ai}
&=a_a^\dagger a_i-a_i^\dagger a_a,\\
T_{abji}
&=a_a^\dagger a_b^\dagger a_j a_i
-a_i^\dagger a_j^\dagger a_b a_a.
\end{aligned}
\end{equation}
Here $i,j$ index occupied spin orbitals and $a,b$ index unoccupied spin orbitals of the Hartree--Fock reference.
Both generators are anti-Hermitian, so $e^{\theta T}$ is unitary for real $\theta$.
Note that each equation in \eqref{eq:gens} contains equally many creation and annihilation operators, hence $[N,T]=0$.
By Proposition~\ref{prop:preserve}, any product of these unitaries acting on a reference in $\Hphys$ satisfies $\Lambda=0$ for every parameter choice.

\subsection{Particle-number-preserving ansatze}

Let $\mathcal P_{\mathrm{UCC}}$ denote the pool of single- and double-excitation generators in \eqref{eq:gens}.

The following families use products of excitation unitaries, with an independent real parameter for each generator.

\begin{itemize}
\item
\textbf{UCCSD}: uses every generator in $\mathcal P_{\mathrm{UCC}}$.

\item
\textbf{SelectedUCC}$_m$: uses the $m$ highest-scoring generators in $\mathcal P_{\mathrm{UCC}}$ under the one-shot score \eqref{eq:oneshot}.

\item
\textbf{RandomUCC}$_m$: uses $m$ distinct generators sampled uniformly without replacement from $\mathcal P_{\mathrm{UCC}}$.
At the same generator budget, this provides a random-selection baseline for SelectedUCC.

\item
\textbf{PCSD}: uses $m$ pair-double generators $T_{a\bar a\bar\imath i}$, repeated over $L$ layers with independent parameters, giving $mL$ parameters.
Here $i,\bar\imath$ and $a,\bar a$ denote spin-up and spin-down spin orbitals of occupied and unoccupied orbitals respectively.
These generators are special cases of \eqref{eq:gens}.
\end{itemize}

For all four families, applying the ansatz to Hartree--Fock initial state $\ket{\psi_{\mathrm{HF}}}\in\Hphys$ produces a state with $\Lambda=0$ for every parameter choice.

\subsection{Two-local ansatz without particle-number conservation}\label{sec:phealeak}

The Two-Local $\mathrm{pHEA}_{XY}$ family consists of single-qubit Pauli rotations generated by $X_q$
and $Y_q$, and two-qubit Pauli rotations generated by $X_pX_q$ and $Y_pY_q$ on adjacent qubits
\begin{equation}\label{eq:phea}
U_{\pHEA}(\theta)=
\prod_{\ell=1}^{L}
\left[
\prod_q
e^{-i\alpha_{\ell q}X_q/2}
e^{-i\beta_{\ell q}Y_q/2}
\prod_{(p,q)\in\mathcal E}
e^{-i\gamma_{\ell pq}X_pX_q/2}
e^{-i\delta_{\ell pq}Y_pY_q/2}
\right].
\end{equation}
Here $\ell=1,\ldots,L$ where $L$ is the number of layers, indexes the circuit layers, $q$ indexes the qubits, and $\mathcal E=\{(q,q+1):q=1,\ldots,n-1\}$.
The real parameters $\alpha_{\ell q}$ and $\beta_{\ell q}$ are the rotation angles for
$X_q$ and $Y_q$, respectively, while $\gamma_{\ell pq}$ and $\delta_{\ell pq}$
are the rotation angles for $X_pX_q$ and $Y_pY_q$ on edge $(p,q)$.
All angles are independently parameterized and
collectively form the variational parameter
vector $\theta$.

Under the Jordan--Wigner encoding, $N=\sum_{q=1}^{n}(I-Z_q)/2$.
The generators $X_q$, $Y_q$, $X_pX_q$, and $Y_pY_q$ do not commute with $N$.
Consequently, the ansatz does not preserve $\Hphys$ for all parameter values, and a state
$U_{\pHEA}(\theta)\ket{\psi_{\mathrm{HF}}}$ may have nonzero leakage even though $\ket{\psi_{\mathrm{HF}}}\in\Hphys$.

\subsection{Raw, projected and penalized objectives}\label{sec:objectives}
For the normalized circuit state $\ket{\psi(\theta)}=U(\theta)\ket{\psi_{\mathrm{HF}}}$, let $p(\theta)=\|P\ket{\psi(\theta)}\|^2=1-\Lambda(\theta)$. The three training objectives are
\begin{alignat}{2}
&F_{\mathrm{raw}}(\theta)&&=\bra{\psi(\theta)}H\ket{\psi(\theta)},\label{eq:obj-raw}\\
&F_{\mathrm{proj}}(\theta)&&=\frac{\bra{\psi(\theta)}PHP\ket{\psi(\theta)}}{p(\theta)},\qquad p(\theta)>0,\label{eq:obj-proj}\\
&F_\mu(\theta)&&=F_{\mathrm{proj}}(\theta)+\mu\Lambda(\theta),\qquad p(\theta)>0,\quad\mu\ge0.\label{eq:obj-penalty}
\end{alignat}
The optimizer minimizes the chosen function, starting from a parameter vector. The Hartree--Fock initial state remains fixed.  The circuit prepares $\ket{\psi(\theta)}$ for all three objectives. Projection is evaluated on the saved statevector, not implemented as an additional circuit operation.

At every iteration with $p(\theta)>0$, both raw and projected energy can be evaluated. Raw energy belongs to the full circuit state and projected energy belongs to $P\ket{\psi(\theta)}/\sqrt{p(\theta)}$. Thus a projected-energy curve for a circuit which is optimized under raw energy is a diagnostic. 

For $0<\Lambda<1$, since $H$ is block diagonal, we have
\begin{equation}\label{eq:raw-decomposition}
F_{\mathrm{raw}}=(1-\Lambda)F_{\mathrm{proj}}+\Lambda E_Q,\qquad
E_Q=\frac{\bra{\psi(\theta)}QHQ\ket{\psi(\theta)}}{\Lambda}.
\end{equation}
Raw energy therefore depends on the energies and probabilities of both components. 
Projected energy contains no direct penalty for outside-sector probability. 
The penalty in \eqref{eq:obj-penalty} uses projected energy which differs from the expectation of $H+\mu(I-P)$ considered by Yen et al.\ \cite[Eq.~(13)]{yen2019projectors}, which is $F_{\mathrm{raw}}+\mu\Lambda$. 
Related penalty methods are discussed in Ref.~\cite{kuroiwa2021}.

\subsection{SelectedUCC ranking score}\label{sec:rules}

Selection of reduced UCC excitation sets has been
studied through energy-based operator sorting \cite{fan2021energysorting} and the unitary selective coupled-cluster method \cite{fedorov2022uscc}.
Here, SelectedUCC$_m$ denotes the specific selection procedure defined below.

For an anti-Hermitian generator $T$ and a normalized state $\ket{\psi}$, we define
\begin{equation}\label{eq:Epsi}
E_\psi(t)
=\braket{\psi(t)|H|\psi(t)},
\qquad
\ket{\psi(t)}=e^{tT}\ket{\psi},
\qquad t\in\mathbb R.
\end{equation}
Because $T$ is anti-Hermitian, $\ket{\psi(t)}$ remains normalized.

SelectedUCC evaluates every candidate once at the Hartree--Fock initial state and retains the $m$ highest-scoring generators.
Writing $E_{\mathrm{HF}}(t)=E_{\psi_{\mathrm{HF}}}(t)$ and $\varepsilon_1=5\times10^{-2}$, the score is
\begin{equation}\label{eq:oneshot}
\begin{aligned}
s(T)
={}&
\underbrace{
\left|
\frac{
E_{\mathrm{HF}}(\varepsilon_1)
-E_{\mathrm{HF}}(-\varepsilon_1)}
{2\varepsilon_1}
\right|
}_{\text{finite-difference gradient term}}
\\
&+
\underbrace{
\frac{
\max\left\{
E_{\mathrm{HF}}(0)-E_{\mathrm{HF}}(\varepsilon_1),
E_{\mathrm{HF}}(0)-E_{\mathrm{HF}}(-\varepsilon_1),
0
\right\}}
{\varepsilon_1}
}_{\text{finite-step energy-decrease term}}.
\end{aligned}
\end{equation}

As $\varepsilon_1\to0$, the first term converges to
\[
\left|
\braket{\psi_{\mathrm{HF}}|[H,T]|\psi_{\mathrm{HF}}}
\right|,
\]
the energy-gradient magnitude used in ADAPT-VQE selection \cite{grimsley2019}, evaluated here at the Hartree--Fock initial state.
The second term is positive when at least one of the two trial steps lowers the energy.
Unlike ADAPT-VQE, this procedure does not repeat the selection during the optimization.

\subsection{Excitation-plus-commutator pool}
\label{sec:enrichedpool}

Lie-algebraic constructions have been studied in VQE generator
sets \cite{izmaylov2020order}.
Here we use a finite excitation-plus-commutator
pool, without requiring commutator closure for the whole Lie algebra.

We extend $\mathcal P_{\mathrm{UCC}}$ by including commutators of excitation generators. 
For $A,B\in\mathcal P_{\mathrm{UCC}}$, define
\begin{equation}\label{eq:commgen}
C_{AB}=[A,B]=AB-BA.
\end{equation}
Let $\mathcal C$ be a collection of nonzero
commutators of this form. 
The enriched candidate pool is
\begin{equation}\label{eq:enrichedpool}
\mathcal P_{\mathrm{enr}}=
\mathcal P_{\mathrm{UCC}}\cup\mathcal C.
\end{equation}
The choice of $\mathcal C$ is specified in the numerical section. 

If $A$ and $B$ are anti-Hermitian, then $[A,B]$ is anti-Hermitian. 
Since the excitation generators commute with $N$,
\begin{equation}
[N,[A,B]]=[[N,A],B]+[A,[N,B]]=0.
\end{equation}
Thus every generator in $\mathcal P_{\mathrm{enr}}$
commutes with $N$. 
Applying a product of their exponential unitaries to a state in $\Hphys$ therefore produces a state with $\Lambda=0$.

\subsection{Tangent-residual ranking at the Hartree--Fock state}
\label{sec:alternative-rules}

Selection of reduced UCC excitation sets has been
studied through energy-based operator sorting
\cite{fan2021energysorting} and the unitary selective
coupled-cluster method \cite{fedorov2022uscc}.

To distinguish energy ranking from selection of state directions, we define
\begin{align}
v_T(\psi)&=(I-\ket\psi\bra\psi)T\ket\psi,\qquad
w_T(\psi)=\begin{pmatrix}\operatorname{Re}v_T(\psi)\\\operatorname{Im}v_T(\psi)\end{pmatrix},\\
S_k(\psi)&=\operatorname{span}_{\mathbb R}\{w_A(\psi):A\in\mathcal A_k\},\qquad
r_k(T;\psi)=(I-\Pi_{S_k(\psi)})w_T(\psi).
\end{align}
Here $\mathcal A_k$ is the selected list and $\Pi_{S_k}$ is its real orthogonal projector. 
The vector $v_T$ removes the component parallel to the state and $r_k$ is the part of the candidate tangent outside the selected span. 
All rules in this subsection evaluate these quantities at the fixed Hartree--Fock initial state.

Here we introduce SelectedUCC, LieBase, LieComm, and
HamiltonianLie selection rules.

\begin{itemize}
\item \textbf{LieBase} selects the excitation candidate with largest $\|r_k\|_2$. 
It favors the largest remaining tangent component outside the selected span, without evaluating its energy benefit.

\item \textbf{LieComm} uses the same residual-norm score on $\mathcal P_{\mathrm{enr}}$. 
It tests whether commutator candidates change this geometric selection. 
\item \textbf{HamiltonianLie} combines normalized energy score with the relative residual,
\begin{equation}\label{eq:hamiltonianlie}
h_k(T)=\widetilde s(T)+\eta\frac{\|r_k(T;\psi_{\mathrm{HF}})\|_2}{\|w_T(\psi_{\mathrm{HF}})\|_2+\delta},
\end{equation}
where $\widetilde s(T)=s(T)/s_{\max}$ if $s_{\max}=\max_{\mathcal P_{\mathrm{enr}}}s>0$, and zero otherwise. 
The first term favors the energy-score criterion and the second term favors a tangent component not represented by the selected span.
$\eta$ is its weight.
\end{itemize}

Residual vectors are not normalized before the LieBase/LieComm comparison, so their scores depend on tangent magnitude as well as direction. 
Once the requested list is selected, it is fixed and its parameters are optimized jointly.

\subsection{Adaptive hybrid selection and optimization}
\label{sec:adaptive-hybrid}

Adaptive Hybrid alternates generator selection and joint parameter optimization in the style of ADAPT-VQE \cite{grimsley2019}, using the
hybrid score defined below.

Let $\ket{\psi_k}$ denote the normalized optimized state of the ansatz containing $k$ selected generators with
$\ket{\psi_0}=\ket{\psi_{\mathrm{HF}}}$.
For an anti-Hermitian candidate generator $T$, define
$$ E_{\psi_k,T}(t) =
\bra{\psi_k}e^{-tT}He^{tT}\ket{\psi_k},
\qquad t\in\mathbb R.$$
The candidate score is
\begin{equation}\label{eq:adaptive-score}
a_k(T) =
\left|
\frac{E_{\psi_k,T}(\varepsilon_a)
      -E_{\psi_k,T}(-\varepsilon_a)}
     {2\varepsilon_a}
\right|
+
\kappa
\frac{\|r_k(T;\psi_k)\|_2}
     {\|w_T(\psi_k)\|_2+\delta}.
\end{equation}
At each selection step, an eligible generator maximizing $a_k(T)$ is appended to the ansatz with its new angle initialized to zero.
All angles are then optimized jointly to obtain $\ket{\psi_{k+1}}$.

The energy term estimates a slope at the current state, while the residual term rewards a tangent component outside the span of previously selected generators evaluated at that same state. Unlike \eqref{eq:oneshot}, this score has no finite-step energy-decrease term. 

The residual basis is rebuilt after each joint optimization. Selection stops at the budget or when no eligible candidate remains, followed by final parameter optimization.

\section{Numerical experiments}\label{sec:numerics}

We evaluate projected-sector error, leakage, and parameter count using noiseless statevector simulations. No finite-shot or hardware results are reported.

\subsection{Molecular systems and numerical quantities}

We use noiseless statevectors, STO-3G Hamiltonians and the Jordan--Wigner transformation. 
PennyLane's differentiable Hartree--Fock backend constructs the Hamiltonians \cite{bergholm2018}. 
Independent PySCF RHF/CASCI calculations \cite{sun2018} check the active-space ground energies.
Each $\EFCI$ is obtained by exact diagonalization restricted to the stated particle-number sector. 
Table~\ref{tab:systems} lists the molecule and target particle-number sectors.
For H$_6$, BeH$_2$ and LiH, the lowest doubly occupied spatial orbital is frozen and the next four spatial orbitals are retained.

\begin{table}[htbp]
\centering\small
\begin{tabular}{@{}lrrrrr@{}}\toprule
System & Distance (\AA) & $n$ & $N_e$ & $d_{\mathrm{phys}}$ & $\EFCI$ [Ha]\\\midrule
H$_4$ chain &1.400&8&4&70&$-2.0290705$\\
H$_4$ stretched &2.400&8&4&70&$-1.8746516$\\
H$_6$ $(4e,4o)$ &1.400&8&4&70&$-2.9359223$\\
BeH$_2$ $(4e,4o)$ &1.326&8&4&70&$-15.5662352$\\
LiH $(2e,4o)$ &1.600&8&2&28&$-7.8636816$\\\bottomrule
\end{tabular}
\caption{Molecular systems. Distances are adjacent-atom separations for the linear chains and bond lengths for BeH$_2$ and LiH. $(ae,bo)$ denotes $a$ active electrons in $b$ active spatial orbitals. $d_{\mathrm{phys}}=\binom n{N_e}$ is the dimension of the target particle-number sector.}\label{tab:systems}
\end{table}

Every run starts from the same Hartree--Fock state for its molecule. For each molecule, method and budget, parameter seeds are $0$ -- $199$. For RandomUCC, seed $r+10000m$ selects an $m$-element excitation subset for parameter seed $r$.

All optimizations use L-BFGS-B \cite{byrd1995} with two-point numerical derivatives and a relative finite-difference step of $10^{-4}$.
Table~\ref{tab:controls} lists the settings for the nonadaptive evaluations and the generator-budget comparison.

For each molecule and generator budget, adaptive selection and parameter optimization produce
one generator sequence and an optimized parameter vector.
The sequence is then held fixed for $200$ optimization runs, each initialized by independently perturbing that parameter vector. 

\begin{table}[htbp]
\centering
\small
\caption{Optimization controls.
A: common nonadaptive protocol.
B0: one-shot and random controls in the generator-budget comparison.
B1: optimization after each adaptive selection
and optimization of the completed sequence.
B2: optimization of the fixed adaptive sequence.}
\label{tab:controls}
\begin{tabular}{@{}lrrrr@{}}
\toprule
Setting & A & B0 & B1 & B2 \\
\midrule
Iteration cap
& 400 & 280 & 90 / 320 & 320 \\
Function-evaluation cap
& 200000 & 15000 & 15000 & 15000 \\
\texttt{ftol}
& $10^{-10}$ & $10^{-11}$ & $10^{-10}$ & $10^{-11}$ \\
\texttt{gtol}
& $10^{-6}$ & $10^{-7}$ & $10^{-7}$ & $10^{-7}$ \\
Maximum line-search steps
& 30 & 40 & 30 & 40 \\
\bottomrule
\end{tabular}

\par\smallskip
\parbox{\linewidth}{
A and B0 initialize each angle independently from $\mathcal N(0,10^{-4})$.
In B1, each newly appended angle is initialized at zero.
The iteration caps are $90$ after each selection and $320$ after completing the sequence.
B2 adds independent $\mathcal N(0,10^{-8})$ perturbations to the stored optimized angles before optimization.
}
\end{table}

All runs with finite final results are included, including $97$ runs under protocol A and $13$ runs in the
generator-budget comparison for which the optimizer reported unsuccessful termination.
Of these, five and four, respectively, reached the iteration cap and none reached the function-evaluation cap.
Successful termination indicates that the optimizer's stopping criterion was satisfied. 
It does not certify a global minimum.

Statevectors use \texttt{complex128} to calculate.
Means and sample standard deviations are computed across optimization runs. 
Shaded regions and bars indicate minimum--maximum ranges.
Leakage in the figures is calculated directly as
$$\Lambda(\psi)=\langle\psi|Q|\psi\rangle,$$
by summing the computational basis probabilities outside the target particle-number sector.

The penalty is $\mu\min\{1,\max\{0,1-p\}\}$, which equals $\mu\Lambda$.
We define chemical accuracy by $\epsphys<1.6\times10^{-3}$~Ha.

The numerical environment uses Python 3.10.12, PennyLane 0.42.3, PySCF 2.6.2, NumPy 2.2.4, and SciPy 1.14.1.

\paragraph{Selection settings.}
The one-shot step is $\varepsilon=0.05$. 
The enriched list uses commutators of the first $\min(20,|\mathcal P_{\mathrm{UCC}}|)$ energy-ranked excitations, keeping at most $120$ nonzero commutators ranked by \eqref{eq:oneshot}. 
For the fixed-state geometric rules, we use $\eta=0.55$, $\delta=10^{-12}$, and candidates with residual norm below $10^{-9}$ are excluded.
We use the Gram--Schmidt process to build the residual basis.
The weight $0.55$ is a fixed experimental choice. 
Adaptive selection uses $\varepsilon_a=0.025$, $\kappa=10^{-3}$ Ha and $\delta=10^{-12}$.
Basis and candidate residuals of norm at most $10^{-10}$ are discarded. 
If the Score ties, we keep the first encountered candidate. 

\subsection{Projected accuracy and leakage}\label{sec:diag}
In the H$_4$ chain diagnostic (Fig.~\ref{fig:diagnostic}), UCCSD and SelectedUCC$_{12}$ have mean errors $9.49\times10^{-4}$ and $9.57\times10^{-4}$ Ha with zero leakage. 
At eight parameters, SelectedUCC has mean error $3.78\times10^{-3}$ Ha, versus $9.81\times10^{-2}$ Ha for RandomUCC.
The six preserving ansatze have zero leakage at every recorded iteration in all $1,200$ runs.

For the 30-parameter $\pHEA$, raw-energy optimization gives mean projected error $0.1413$ Ha and mean leakage $2.17\times10^{-10}$.
Projected-energy optimization gives a smaller mean projected error, $8.18\times10^{-2}$ Ha, but mean leakage $0.9574$. 
Its mean raw energy is $-0.3472$ Ha, compared with $-1.8878$ Ha after raw-energy optimization: lower projected energy is not lower raw energy. 
The $\mu=10$ Ha penalty gives mean leakage $9.25\times10^{-12}$ and mean projected error $0.1413$ Ha. 
The $\mu=1,100$ Ha results similarly have leakage below $3.33\times10^{-11}$ and projected errors near $0.1413$ Ha.
None reaches chemical accuracy. 
Relative to projected-only optimization, the penalized outcomes have lower leakage but larger projected error and relative to raw-energy optimization, they show no projected-error improvement.

\begin{figure}[htbp]
\centering
\includegraphics[width=0.92\linewidth]{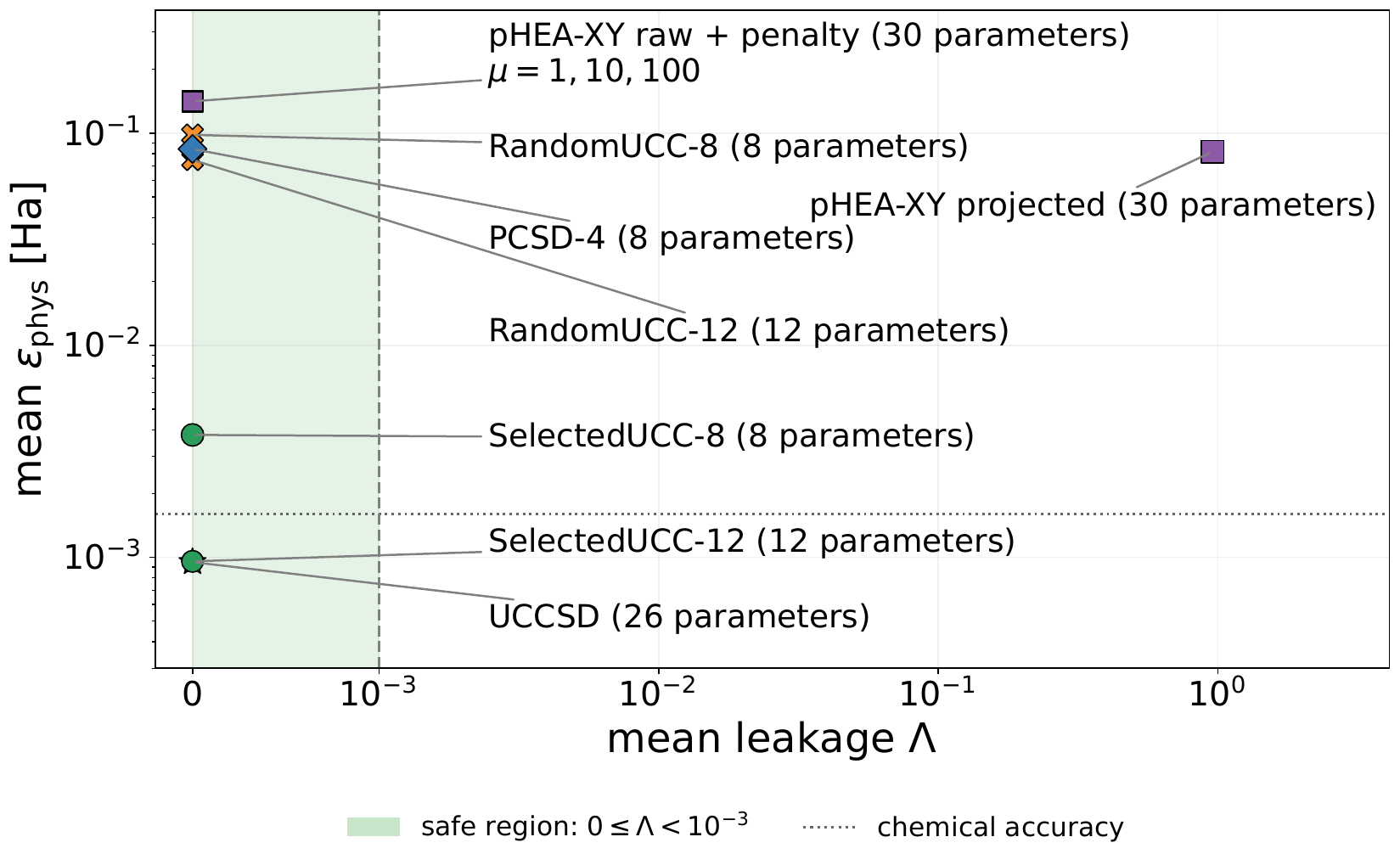}
\caption{H$_4$ chain final-state diagnostic: means over $200$ optimizer seeds per ansatz and objective. Labels give parameter counts. The chosen leakage interval $0\leq\Lambda<10^{-3}$ is shaded. The dotted horizontal line is chemical accuracy.}\label{fig:diagnostic}
\end{figure}

For BeH$_2$ and LiH, the circuit uses two layers, each applying $R_Y$ rotations to all qubits, followed by $XX$ rotations on all adjacent  pairs and then $YY$ rotations on all adjacent pairs. 
This circuit has $44$ parameters.

For BeH$_2$ and LiH molecules, the experiment gives projected energies within chemical accuracy in all $200$ runs on each molecule, while mean leakage is $39.8\%$ and $42.6\%$ (Table~\ref{tab:additional}).
All $400$ optimizations terminate successfully and their maximum projected errors are $1.07\times10^{-4}$ and $8.43\times10^{-5}$ Ha, respectively. 
These are examples of chemically accurate projected energies with substantial leakage, unlike the H$_4$ results.

\begin{table}[htbp]
\centering\small
\begin{tabular}{@{}lcc@{}}\toprule
System & $\epsphys$ [Ha] & $\Lambda$\\\midrule
BeH$_2$ $(4e,4o)$ & $(2.35\pm7.46)\times10^{-6}$ & $0.398\pm0.291$\\
LiH $(2e,4o)$ & $(0.825\pm6.08)\times10^{-6}$ & $0.426\pm0.303$\\\bottomrule
\end{tabular}
\caption{Optimization using $F_{\mathrm{proj}}$, $200$ optimizer seeds per molecule. Values are mean $\pm$ sample standard deviation, not from measurement shots.}\label{tab:additional}
\end{table}

\subsection{Optimization Energy trajectories}\label{sec:traj}
Figure~\ref{fig:trajectory} reports projected energy, raw energy and leakage along the same parameter trajectories. 
UCCSD and SelectedUCC$_8$ have zero leakage at every iteration, so their raw and projected energies coincide and their final energies are lower than their initial energies.

For the 30-parameter $\pHEA$ optimized using $F_{\mathrm{proj}}$, mean projected energy decreases from $-1.88779$ to $-1.94731$ Ha, while mean raw energy increases from $-1.88734$ to $-0.34718$ Ha. 
Mean leakage increases from $6.80\times10^{-4}$ to $0.9574$. 
Even though projected energy decreases in every run, but final projected errors remain $0.08132$--$0.08443$ Ha, which is outside of chemical accuracy. 
Final leakage spans $0.6697$--$0.9944$ with $191$ runs exceed $0.93$, and five reach the iteration cap. Thus this H$_4$ experiment demonstrates reduced projected error with increased outside-sector probability, not chemically accurate projected energy.
The BeH$_2$ and LiH results are in Table~\ref{tab:additional}.

The nearly leakage-free endpoints obtained using $F_{\mathrm{raw}}$ and $F_\mu$ with $\mu=10$ Ha instead have energies near $E_{\mathrm{HF}}=-1.8877903066$ Ha, above $\EFCI=-2.0290704961$ Ha.
Their final raw and projected energies, respectively, differ from $E_{\mathrm{HF}}$ by at most $2.46\times10^{-9}$ and $1.76\times10^{-8}$ Ha. 
This is agreement in energy, not a claim that their states equal the Hartree--Fock state.
Neither the Hartree--Fock initial state nor the observed small leakage establishes ground-state preparation.

\begin{figure}[htbp]
\centering
\includegraphics[width=\linewidth,keepaspectratio]{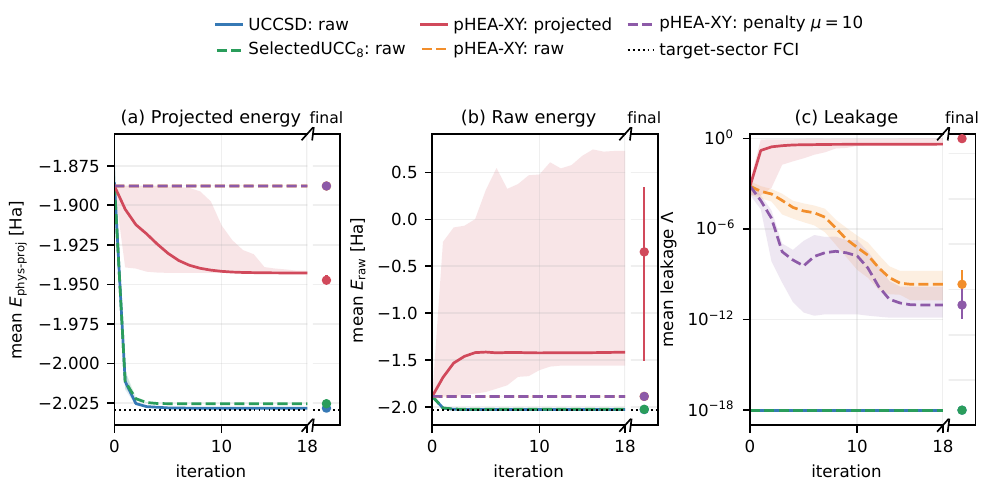}
\caption{H$_4$ chain trajectories, $200$ seeds per method: (a) projected energy, (b) raw energy and (c) leakage. pHEA labels identify the optimized objective \eqref{eq:obj-raw}--\eqref{eq:obj-penalty}. Curves show means and shading/final bars show minimum--maximum ranges. The penalty uses $\mu=10$ Ha. The $\EFCI$ line denotes the target-sector ground energy, not necessarily the minimum over all particle-number sectors. Leakage uses $10^{-18}$ as the display floor.}\label{fig:trajectory}
\end{figure}

\subsection{Fixed-generator budget comparisons}\label{sec:budget}
Figure~\ref{fig:budget} compares fixed-generator budgets $m\in\{2,4,6,8,10,12,16\}$ on the two H$_4$ geometries and active-space H$_6$.
In every tested molecule and budget pair, SelectedUCC, LieBase, LieComm and HamiltonianLie select the same ordered generators. 
Their final errors agree exactly across the $200$ seeds. 
In these selections, geometric-score ties reproduce the same sequence hence the curves do not establish an independent geometric improvement.

SelectedUCC has lower mean error than RandomUCC at $m=16$ on all three systems (Table~\ref{tab:budget}). This ordering does not hold at every budget.
On stretched H$_4$ with $m=2$, RandomUCC has mean error $0.5006$ Ha, but $0.5149$ Ha for SelectedUCC.
At $m=8$, stretched-H$_4$ SelectedUCC has mean error $0.01820$ Ha and standard deviation $0.04494$ Ha. 
Its mean error is not monotone in the generator budget.

\begin{table}[htbp]
\centering\small
\begin{tabular}{@{}lcc@{}}\toprule
System & SelectedUCC & RandomUCC\\\midrule
H$_4$ chain & $9.49\times10^{-4}\ (1.25\times10^{-10})$ & $4.90\times10^{-2}\ (2.45\times10^{-2})$\\
H$_4$ stretched & $7.64\times10^{-3}\ (2.07\times10^{-6})$ & $2.07\times10^{-2}\ (2.94\times10^{-2})$\\
H$_6$ $(4e,4o)$ & $2.32\times10^{-4}\ (1.06\times10^{-10})$ & $3.79\times10^{-2}\ (1.81\times10^{-2})$\\\bottomrule
\end{tabular}
\caption{Common-protocol comparison at $m=16$: mean projected-sector error, with sample standard deviation in parentheses. $200$ optimizer seeds per entry. Unsuccessful finite endpoints are also included.}\label{tab:budget}
\end{table}

\begin{figure}[htbp]
\centering
\includegraphics[width=\linewidth,keepaspectratio]{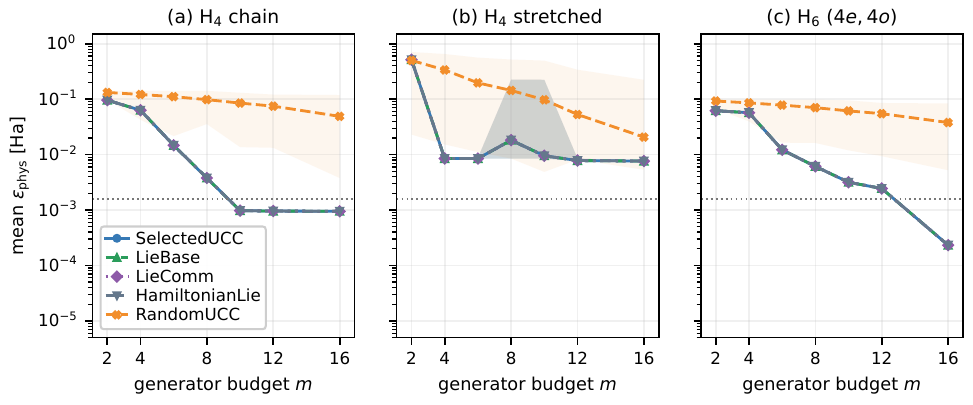}
\caption{Common-protocol fixed-generator comparisons on all three budget-test systems. Curves show means over $200$ seeds per method and budget, with minimum--maximum shading. Dotted horizontal line is chemical accuracy. The four deterministic rules select identical ordered generators.}\label{fig:budget}
\end{figure}

\subsection{Adaptive hybrid comparison}
\label{sec:adaptive-results}

We compare Adaptive Hybrid, one-shot SelectedUCC, and RandomUCC using $200$ final optimization runs
per method and generator budget.
The optimizer settings are B0--B2 in Table~\ref{tab:controls}.

Adaptive Hybrid constructs the ansatz iteratively that each generator is selected using the current state, appended to the ansatz, and followed by joint parameter optimization before the next selection.
For each molecule and generator budget, this construction is performed once.
The resulting generator sequence and order are then held fixed during $200$ final optimization runs, each initialized by perturbing the stored optimized parameters.

Figure~\ref{fig:adaptive} shows the comparison for three molecular systems.
Adaptive Hybrid achieves a lower mean projected-sector error than one-shot SelectedUCC at all five entries in Table~\ref{tab:adaptive}.
At some other budgets the differences are smaller: for the H$_4$ chain, the mean errors nearly
coincide at $m=6,8,10$.

\begin{table}[htbp]
\centering
\small
\begin{tabular}{@{}lrcc@{}}
\toprule
System & $m$ & One-shot SelectedUCC & Adaptive Hybrid \\
\midrule
H$_4$ chain
& 4 & $6.31\times10^{-2}$ & $2.16\times10^{-2}$ \\
H$_4$ chain
& 16 & $9.49\times10^{-4}$ & $1.45\times10^{-4}$ \\
H$_4$ stretched
& 2 & $5.15\times10^{-1}$ & $2.30\times10^{-2}$ \\
H$_6$ $(4e,4o)$
& 4 & $5.70\times10^{-2}$ & $1.71\times10^{-2}$ \\
H$_6$ $(4e,4o)$
& 16 & $2.32\times10^{-4}$ & $1.20\times10^{-4}$ \\
\bottomrule
\end{tabular}
\caption{
Mean projected-sector error [Ha] over $200$ final optimization seeds per entry.
As the entries shown, sample standard deviations do not exceed $8.53\times10^{-10}$~Ha for one-shot
SelectedUCC and $1.74\times10^{-9}$~Ha for Adaptive Hybrid.
}\label{tab:adaptive}
\end{table}

\begin{figure}[htbp]
\centering
\includegraphics[width=\linewidth,keepaspectratio]
{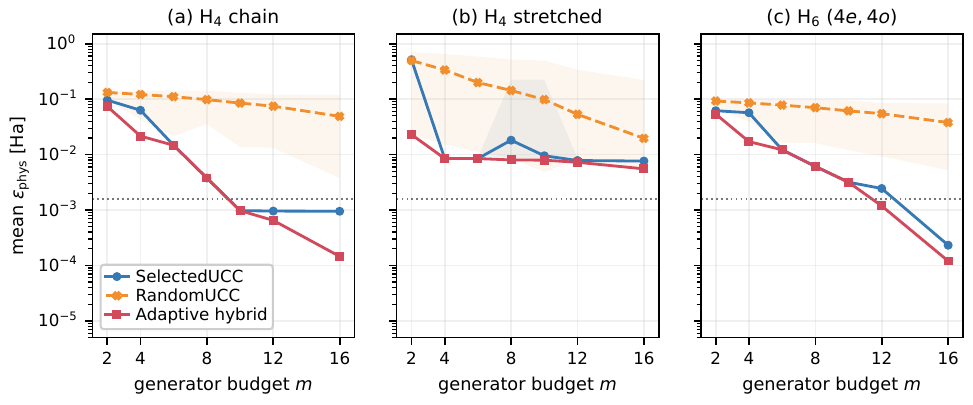}
\caption{
Adaptive Hybrid (red), one-shot SelectedUCC, and RandomUCC.
Curves show mean projected-sector error over
$200$ final optimization runs per method and generator budget. 
Shading indicates minimum--maximum ranges.
For Adaptive Hybrid, the selected generator sequence is held fixed across optimization runs.
Dotted lines indicate the chemical-accuracy.
}
\label{fig:adaptive}
\end{figure}

\section{Discussion and limitations}\label{sec:discussion}

The simulations distinguish projected-sector accuracy from particle-number preservation.
On BeH$_2$ and LiH, chemically accurate projected energies might exist with substantial leakage.
Number-preserving generators eliminate leakage induced by the ansatz, but do not guarantee a sufficiently expressive ansatz or successful energy minimization.

The generator comparisons are narrower. SelectedUCC has lower mean errors than random subsets at $m=16$ on all three tested molecules, whereas the tested geometric rules supply no distinct selected sequence. 

These are small, noiseless, minimal-basis and active-space tests. 
They are not confidence intervals or measurement error. 
Parameter count does not determine gate count or noise sensitivity as well. 
The diagnostic error also needs a known $\EFCI$, so it is a benchmark quantity rather than an independently available certificate for an unsolved molecule. 
Channel identities in Section~\ref{sec:channel} describe noise analytically. There is no noisy-device or finite-shot experiment being conducted.

\section{Conclusion}\label{sec:conclusion}
Projected energy accuracy and target-sector occupation are distinct properties of a VQE state.
Reporting $(\epsphys,\Lambda,\#\theta)$ makes this distinction explicit while recording ansatz size. 
The BeH$_2$ and LiH benchmarks show that projected energies within chemical accuracy can accompany substantial outside-sector probability, whereas number-preserving excitation circuits retain zero leakage during optimization. 
Energy ranking gives lower mean errors than random subsets at $m=16$ on the three budget-test systems, but the geometric-rule agreement and the separately optimized adaptive results do not establish an independent geometric or adaptive advantage. 
Energy accuracy should therefore be assessed together with sector probability, not used as a substitute for it.

\bibliographystyle{abbrv}

\end{document}